\documentclass[10pt,letterpaper]{IEEEtran}
\usepackage[T1]{fontenc}
\usepackage[utf8]{inputenc}
\usepackage{authblk}
\usepackage{amsthm}
\usepackage{cite}
\include{psfig}
\usepackage{psfig}
\usepackage{psfrag}
\usepackage{amssymb}
\usepackage{picins}
\usepackage[demo]{graphicx}
\graphicspath{ {images/} }
\usepackage{wrapfig}
\usepackage{graphicx}
\usepackage{tikz}
\usepackage{multicol,lipsum}
\usepackage{subcaption}
\usepackage{enumitem}
\usepackage{tikz}
\usetikzlibrary{calc}
\usepackage[tbtags]{amsmath}
\usepackage{epsfig}
\usepackage{epstopdf}
\usepackage{stfloats}
\usepackage{float}
\usepackage{array}
\usepackage{color}
\usepackage{pgfplots}
\usepackage{textcomp}
\usepackage{stmaryrd}
\usepackage{gensymb}
\usepackage{multirow}
\usepackage{url}
\usepackage{comment}
\usepackage[ruled,vlined]{algorithm2e}
\usepackage{algorithmic}
\usepackage{adjustbox}
\usepackage{lipsum}
\pgfplotsset{compat=newest}
\pgfplotsset{plot coordinates/math parser=false}
\usepackage{stmaryrd}

\newtheorem{lemma}{Lemma}
\IEEEoverridecommandlockouts
\newlength\tindent
\renewcommand{\indent}{\hspace*{\tindent}}

\def\BibTeX{{\rm B\kern-.05em{\sc i\kern-.025em b}\kern-.08em
  T\kern-.1667em\lower.7ex\hbox{E}\kern-.125emX}}
  
\newtheorem{definition}{Definition}  
  
\begin{document}
\title{A Downlink Puncturing Scheme for Simultaneous Transmission of URLLC and eMBB Traffic by Exploiting Data Similarity}

\author{Mohammed AL-Mekhlafi, Mohaned Chraiti, Mohamed Amine Arfaoui, \\Chadi Assi, Ali Ghrayeb, and Amira Alloum \thanks{M. AL-Mekhlafi, Mohamed Amine Arfaoui and C. Assi are with the Concordia Institute for Information Systems Engineering, Concordia University, Montreal, QC, Canada
(e-mail: {m$\_$almekh@encs.concordia.ca}; {m$\_$arfaou@encs.concordia.ca}; {assi@ciise.concordia.ca}).}
\thanks{M. Chraiti is with the LIDS lab, Massachusetts Institute of Technology (MIT), Cambridge, Massachusetts, USA
(email: mchraiti@mit.edu.}
\thanks{A. Ghrayeb is with Texas A \& M University at Qatar, Doha, Qatar (e-mail:
{ali.ghrayeb@qatar.tamu.edu}).} 
\thanks{A. Alloum is with Huawei, Paris, France (email: amira.alloum@huawei.com.)}}

	
	\maketitle
	
	\thispagestyle{empty}
	\begin{abstract}
		Ultra Reliable and Low Latency Communications (URLLC) is deemed to be an essential service in 5G systems and beyond (also called 6G) to accommodate a wide range of emerging applications with stringent latency and reliability requirements. Coexistence of URLLC alongside other service categories calls for developing spectrally efficient multiplexing techniques. Specifically, coupling URLLC and conventional enhanced Mobile BroadBand (eMBB) through superposition/puncturing naturally arises as a promising option due to the tolerance of the latter in terms of latency and reliability. The idea here is to transmit URLLC packets (typically sporadic and of short size) over resources occupied by ongoing eMBB transmissions while minimizing the impact on the eMBB transmissions. In this paper, we propose a novel downlink URLLC-eMBB multiplexing technique that exploits possible similarities among URLLC and eMBB symbols, with the objective of reducing the size of the punctured eMBB symbols. We propose that the base station (BS) scans the eMBB traffic' symbol sequences and punctures those that have the highest symbol similarity with that of the URLLC users to be served. As the eMBB and URLLC may use different constellation sizes, we introduce the concept of symbol region similarity to accommodate the different constellations. We assess the performance of the proposed scheme analytically, where we derive closed-form expressions for the symbol error rate (SER) of the eMBB and URLLC services. {We also derive an expression for the eMBB loss function due to puncturing in terms of the eMBB SER}. We demonstrate through numerical and simulation results the efficacy of the proposed scheme where we show that 1) the eMBB spectral efficiency is improved by puncturing fewer symbols, 2) the SER and reliability  performance of eMBB are improved, and 3) the URLLC data is accommodated within the specified delay constraint while maintaining its reliability, 4) and the proposed strategy has polynomial time complexity making it an efficient solution to be used in practice.     
	\end{abstract}
	
	\vspace{-.1 in}
	\section{Introduction}\label{sec:introduction}

\subsection{Motivation}
5G and beyond systems are anticipated to provide a variety of service classes with different requirements in terms of latency, reliability and connectivity \cite{saad2019vision,8808168,bennis_risk}. This naturally raises concerns about their coexistence, especially after it has been shown that allocating a dedicated bandwidth for each service is not spectrally efficient \cite{8403963}. In particular, providing a dedicated bandwidth for ultra reliable and low latency communications (URLLC) class of service has been shown to be poorly efficient where the effectively used bandwidth could be less than 5\% of the total allocated resource. This is mainly due to URLLC traffic characteristics and requirements \cite{8491078}. In fact, the URLLC services come with stringent requirements in terms of latency (less than one millisecond) and reliability (packet error less than $10^{-6}$), implying that such services require immediate availability of spectral resources \cite{maaz2018urllc,bennis_risk}. Meanwhile, given the sporadic characteristic of URLLC traffic and their short packet size, the allocated resources will only be used occasionally and for a short period \cite{saad2019vision,8808168,8472907,park2020extreme}. Therefore, on-demand resource allocation for URLLC transmissions is deemed a promising solution to make good use of spectral resources.\\
\indent Aligning with the concept of on-demand allocation of resources for URLLC communications, the 3GPP standard has suggested using superposition/puncturing for multiplexing URLLC and enhanced Mobile BroadBand (eMBB) services in 5G networks \cite{Meeting87,pedersen2017punctured}. As time is divided into slots, and each slot consists of several mini-slots \cite{7842312}, the main idea of the superposition/puncturing framework is to transmit URLLC packets in mini-slot basis, upon their arrival, over the resources occupied by ongoing service type transmissions. Specifically, the eMBB traffic shares the time-frequency resources within each slot, which can be based on the channel states of the eMBB traffic. To accommodate the URLLC traffic of the tight latency constraints, the arriving URLLC packets are immediately scheduled in the next mini-slot on top of the ongoing eMBB transmissions. If the BS  allocates transmission power for both eMBB and URLLC traffic, then it is referred to as superposition. If the BS chooses zero transmission power for the eMBB traffic, then this is referred to as puncturing \cite{pedersen2017punctured}.

  \subsection{Related Work}

Superposition/puncturing is considered a promising option to allocate the URLLC traffic due to the tolerance of the latter in terms of latency and reliability. Hence, much work \cite{8476595,8638930,8932425,8663990,joint,8643428,8746407,9171342} focused on developing techniques based on coupling URLLC and conventional (eMBB) data transmission through superposition/puncturing.  In \cite{8476595,8638930,8932425,8663990,joint,8643428,8746407,9171342}, the authors investigated and developed novel superposition/puncturing approaches aiming to minimize the impact on eMBB in terms of the contaminated symbols. The advantages of using  superposition for sharing resources in uplink communications between eMBB,
mMTC, and URLLC devices was studied  in \cite{8476595}. For URLLC downlink MIMO-NOMA,  network layer performance bounds
and cross-layer power control were studied in \cite{8638930}.
A max-matching diversity (MMD) algorithm was proposed in \cite{8932425} to allocate eMBB users, considering both heterogeneous orthogonal and non-orthogonal multiple access network slicing strategies. A machine learning approach for hybrid multiple access solution (HMA) was proposed in \cite{8663990}. In fact, the classical methods of NOMA, such as power-domain, require perfect channel state information (CSI) at the BS such that the transmitted signal can be separated at the receiver with successive interference cancellation (SIC)\cite{Chraiti_NOMA}. A new class of NOMA, namely bits similarity NOMA, was proposed in \cite{Chraiti_NOMA}. It was shown that, without perfect CSI, bit-similarity NOMA can achieve better spectral efficiency and fairness among users compared to traditional NOMA techniques. In fact, superposition techniques (NOMA) cause severe degradation to the URLLC reliability because the eMBB signal acts as an interference signal that increases the decoding errors of the URLLC traffic. Moreover, the lack of the URLLC CSI at the transmitter decreases the chances to superpose URLLC traffic on eMBB \cite{8612914}. Thus, puncturing is preferred as it conserves URLLC reliability.

In order to study the impact of puncturing  eMBB resources to accommodate URLLC transmission, the authors in \cite{joint} studied the problem of joint scheduling of eMBB and URLLC data transmission according to linear, convex and threshold models of the eMBB rate loss associated with the eMBB resources puncturing. A risk-sensitive approach was introduced in \cite{8643428} to mitigate the risk of puncturing eMBB resources. A resource allocation scheduler was proposed in \cite{8746407} where the formulated problem considered the overhead associated with the URLLC load segmentation while maximizing the rate utility. A null-space-based spatial puncturing scheduler for joint URLLC/eMBB traffic was proposed in \cite{9187217}. The authors in \cite{manzoor2020contract} formulated a URLLC traffic allocation problem by adopting a superposition or puncturing scheme. Practically, when the URLLC service is initiated in the middle of the eMBB transport block, part of eMBB symbols are replaced by and/or superposed with the symbols of the URLLC packet. Accordingly, the reception quality of the eMBB services could be degraded severely.\\ 
\indent  
Since eMBB tolerates delays, eMBB users can rely on long error-correction codes in combination with re-transmission techniques to compensate for the loss incurred by superposition/puncturing. \textit{Retransmission-based puncturing} slows the eMBB traffic, and it requires more overhead including puncturing indicator (PI) to inform the eMBB user of the punctured resources, while the whole information block can be re-transmitted if decoding errors occur. Therefore, researchers have been thinking about using codes (\textit{code-based puncturing}) to correct the erroneous symbols in the eMBB message and hence avoiding retransmissions and high overhead signal\cite{8403963,pedersen2017punctured}. Particularly, the gain achieved by retransmission based puncturing over \textit{code-based puncturing} is moderate and less than 10\% \cite{8403963,pedersen2017punctured}. Moreover, indicator-free scheme including a transmit precoding with blind detection is proposed for resource overhead reduction \cite{8891466}. In general, the more punctured eMBB symbols there are, the higher the number of erroneous eMBB symbols and the lower the code rate ( of the error correction code) we get, which subsequently results in low spectral efficiency. In this paper, we aim to reduce the number of contaminated eMBB symbols for the uncoded system and hence the possibility to enhance the code rate ( of the error correction code) of the eMBB and then the spectral efficiency.


\subsection{Contributions}
{In this work, we are motivated to satisfy QoS of both eMBB and URLLC services in 5G and beyond 5G systems. Therefore, we seek to develop a puncturing strategy such that the impact on the punctured eMBB symbols is minimized, which should essentially lead to better eMBB QoS and spectral utility. In other words, we aim at devising a puncturing strategy that can decrease the impact of simultaneous transmissions of URLLC and eMBB traffic on the eMBB traffic. Hence, there is no need to inform the eMBB users about punctured resources, i.e., avoid transmitting costly and unnecessary puncturing indicator signal. The contributions of the proposed downlink puncturing strategy are summarized as follows: } 
\begin{itemize}
    \item Different from existing works, we exploit the possible similarity among the URLLC-eMBB symbols instead of random allocation. Indeed, upon the arrival of a URLLC packet, the BS scans the ongoing eMBB transmissions and selects the one that maximizes the number of similar symbols between the two services. In fact, increasing the similarity between the eMBB-URLLC symbols effectively reduces the impacted eMBB symbols and hence the possibility to enhance the used error correction code rate and then the spectral efficiency.
    \item While developing the proposed technique, we consider the  case where an eMBB user could have different symbol constellations than that of URLLC users. Accordingly, we introduce the so-called similarity region to evaluate the similarity between the eMBB and URLLC with different constellations. We describe in detail the encoding and decoding processes for both eMBB and URLLC traffic.
    \item  Taking into consideration the symbol errors occurring due to the channel impairment and the puncturing process, we derive a closed-form expression for the symbol error rate (SER) of the eMBB traffic. The expression shows that the SER of the eMBB traffic depends on the signal-to-noise ratio (SNR), the average URLLC load, and the average similarity.  We also consider the SER to measure the reliability of the URLLC traffic, as conserving the SER of the URLLC preserves the minimum packet error rate. Moreover, other reliability improvement techniques, i.e., error control coding schemes \cite{sharma2017polar}, packet duplication \cite{rao2018packet}, and HARQ \cite{strodthoff2019enhanced}, can be used to enhance the URLLC reliability. These enhancement techniques are outside the scope of this work. 
    \item We demonstrate through several numerical examples the efficacy of the proposed scheme where we show that gains of up to 10$~$dB can be achieved in comparison to the code-based puncturing technique. At high SNR, the eMBB SER is dominated by error occurring due to {puncturing}, i.e., the impact of channel diminishes. The opposite is true when the similarity increases, that is, the SER is greatly affected by the channel, not the {puncturing}. We also show that the proposed algorithm has low complexity computational time making it an efficient  solution to be used in practice.
\end{itemize}

\subsection{Paper Outline}
The rest of the paper is organized as follows. Section \ref{sec:system_model}
describes the adopted system model. The proposed {puncturing} strategy is described in \ref{sec:ProposedSolution}. Section \ref{sec:evaluation} provides performance analysis of the proposed strategy where closed-form expressions for the SER for both eMBB and URLLC users are derived. Numerical and simulations results are shown
in \ref{sec:results}. We conclude the paper in Section \ref{sec:conclusion}.

\section{ System Model}\label{sec:system_model}
\begin{figure}[!t]
	\centering
	\center{\includegraphics[width=1\columnwidth,draft=false]
		{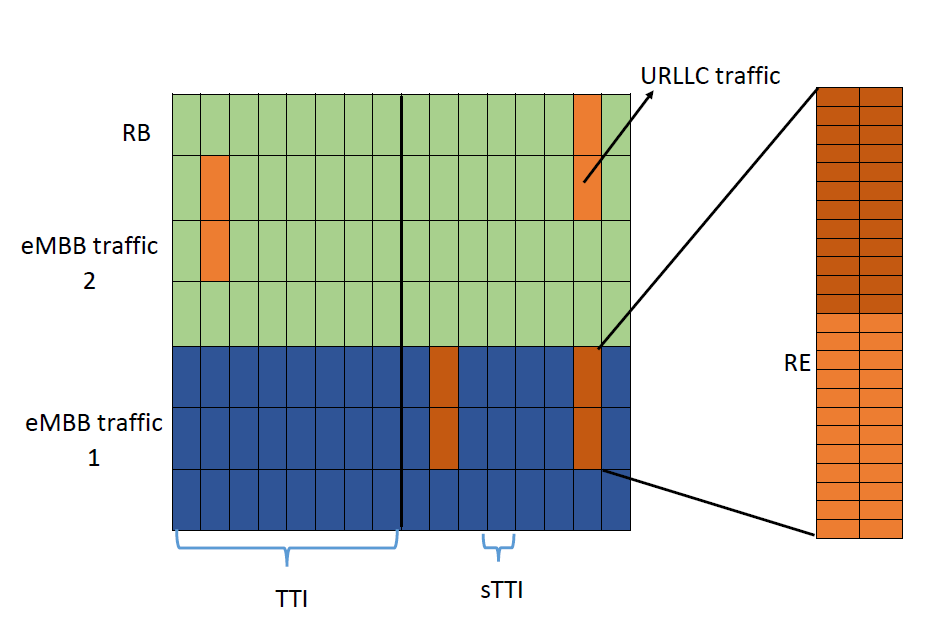}}
	\caption{Relation between frequency resources and {puncturing} mechanism.}
	\label{Fig:spectrum_resources}
	\vspace{-.0 in}
\end{figure} 
We consider a downlink wireless system consisting of one BS that serves certain eMBB and URLLC traffics simultaneously. The system bandwidth is partitioned into $L$ equally sized frequency resources, where each frequency resource is referred to as a resource element (RE). Each $12$ REs constitute a resource block (RB) that is equivalent to $180$ KHz. The time domain is divided into slots, also known as transmission time intervals (TTIs). The duration of each TTI is $1$ ms. To support the low latency requirement of the URLLC traffic, each TTI is further divided into mini-slots, also known as small TTIs (sTTIs), where the duration of each sTTI is $0.143$ ms \cite{joint}. The REs are assigned to the eMBB traffic at the beginning of each TTI, while the URLLC traffic arriving at each sTTI is directly transmitted in the next sTTI by puncturing the REs belonging to the eMBB load. Each URLLC packet is divided into blocks of $\zeta$-symbols, with $\zeta \geq 1$, and it is allocated within one sTTI. Each eMBB receiver is assumed to decode its received data without knowledge of the punctured resources, i.e, each eMBB receiver is assumed to be unaware of the punctured resources at the transmission. Accordingly, the puncturing overhead is reduced. Based on this assumption, each eMBB receiver decodes its received data according to its decoder. Let $m\in\{2, 4, ...,{M}\}$ denote the order of the quadrature amplitude modulation (QAM) scheme adopted for the eMBB traffic with symbol error probability $P_m(\gamma_{\rm e})$, where $\gamma_{\rm e}$ denote the received eMBB signal-to-noise ratio (SNR). In addition, let $n\in\{2, 4, ...,{N}\}$ denote the order of the QAM scheme adopted for the URLLC traffic and let $ \epsilon_{\rm u} $ denote its target symbol error probability.\footnote{By definition, both $m$ and $n$ are powers of $2$.} Practically, the URLLC modulation order $n$ is low due to the following reasons: 1) the lack of accurate channel estimation due to latency constraints; 2) the URLLC traffic is assumed to be small, so the achievable capacity follows the short-block regime; and 3)  the high reliability constraint of the URLLC traffic, which means very low SER. We list in Table \ref{Table:parameters} most of the variables used in the analysis throughout the paper. 
 \begin{table}[!t]
 \caption{\label{Table:parameters} List of Variables Used in the Analysis.}
\centering 
\begin{tabular}{l l}
 \hline
Symbol& Description     \\
 \hline
 $\Omega$& similarity region  \\
 $n$& URLLC modulation order\\
 $m$& eMBB modulation order\\
 $L$& BS downlink frequency resources  \\
 $l$& average URLLC traffic   \\
 $\gamma_{\rm e}$& eMBB Signal to Noise Ratio  \\
  $\gamma_{\rm u}$& URLLC Signal to Noise Ratio \\
 $L_m$& eMBB frequency resources with modulation order $m$  \\
 $l_{n,m}$&  {punctured} eMBB symbols of modulation order $m$ \\
 &by URLLC traffic of modulation order $n$\\
 $\mathcal{L}_{n,m}$&  effectively {punctured} eMBB symbols of modulation order $m$ \\ 
 &by URLLC traffic of modulation order $n$\\
  $\overline{\mathcal{L}}_{n,m}$&  non-effectively {punctured} eMBB symbols of modulation order $m$\\ 
  &by URLLC traffic of modulation order $n$\\
 $p_m$&   the probability of encoding eMBB with modulation order $m$  \\
 $P(.)$& eMBB traffic symbol SER  \\
 $\mathcal{P}(.)$& URLLC traffic symbol SER  \\
 $\zeta$& URLLC block size \\
 $U_{n,m}(.)$& average similar symbols\\
 ${K}$& similarity search space\\
  $s_u$& URLLC symbol  \\
 $s_e$& eMBB symbol  \\
 \hline
\end{tabular}
\end{table}	

\section{Proposed {Puncturing} Scheme}\label{sec:ProposedSolution}
\subsection{Rationale}
The main idea of the proposed puncturing strategy is to exploit the similarity between the symbols of the URLLC block and the symbols of the eMBB load such that the punctured eMBB symbol is similar to the transmitted URLLC symbol. Instead of puncturing the eMBB traffic randomly or greedily, we can search through the eMBB information block to exploit the eMBB sequence that has the highest similarity to the URLLC data block. Fig. \ref{Fig:spectrum_resources} illustrates the mechanism of the proposed scheme in terms of frequency and time resources. At each mini-slot, one can transmit two orthogonal frequency division multiplexing (OFDM) symbols per RE \cite{ghosh20185g}. Hence, if we consider a wireless network with $100$ RBs, then a total of $100 \times 12 \times 2 = 2400$ ODFM symbols can be transmitted in one sTTI. The BS searches for similarity between the URLLC sequence with the ongoing $2400$ eMBB symbols and it allocates the URLLC traffic over the eMBB sequence that has the maximum similarity. For example, assume that the URLLC block length is $2$ RB, i.e., $2\times 12 \times 2 = 48$ OFDM symbols, and that the search window (step) is one RB. Then, the proposed algorithm evaluates the similarity between the URLLC sequence and ${K}=99$ possibilities, where ${K}$ is the search space. As a result, the BS punctures the eMBB sequence that has the maximum similarity to the URLLC sequence.\\
\indent For more elaboration, let us assume that both eMBB and URLLC services employ binary phase shift keying (BPSK) modulation and that the transmitted URLLC symbol is \textbf{0}. Then the punctured eMBB symbol can be either \textbf{0} or \textbf{1}. If the punctured eMBB symbol is \textbf{0}, then the transmitted URLLC symbol and the punctured eMBB symbol are similar, and therefore, and therefore, the error probability of the eMBB symbol is not affected by the puncturing scheme. However, if the punctured eMBB symbol is \textbf{1}, then the eMBB symbol will be received erroneously with probability {($P(0)=1-P(1)= 1-P_2(\gamma_{\rm e})$)}. Therefore, it is recommended to puncture the eMBB traffic that has the maximum similarity to the URLLC traffic. Intuitively, increasing the similarity between the transmitted URLLC symbols and the punctured eMBB symbols will reduce the symbol error rate at the eMBB receiver, which reduces retransmissions and PI overhead. \\
\indent In practice, the modulation schemes used by the eMBB and the URLLC traffics can be different. In addition, the eMBB receiver, which is unaware of the punctured part of the transmission, decodes the received signal using a maximum likelihood receiver. Based on this, the probability of receiving the punctured eMBB symbols in error depends on the Euclidean distance between both the transmitted URLLC symbols and the punctured eMBB symbols. As an illustration, let us consider the case when the URLCC traffic employs BPSK modulation and the eMBB traffic employs 4-QAM modulation and let us suppose that the transmitted URLLC symbol is \textbf{\{0\}}. As shown in Fig.~\ref{Fig:Different_modulation_scheme}, it is preferred to puncture the eMBB symbols \textbf{00} and \textbf{01}, since they have the lowest Euclidean distance to \textbf{0} as compared to the symbols \textbf{10} and \textbf{11}, and therefore, a lower resulting error probability. Specifically, as shown in Fig.~\ref{Fig:Different_modulation_scheme}, we can say that the symbols \textbf{\{0, 00, 10\}} belong to the same region which so-called the similarity region according to the following definition:
    \begin{figure}[t]
	\centering
	\center{\includegraphics[width=.8\columnwidth,draft=false]
		{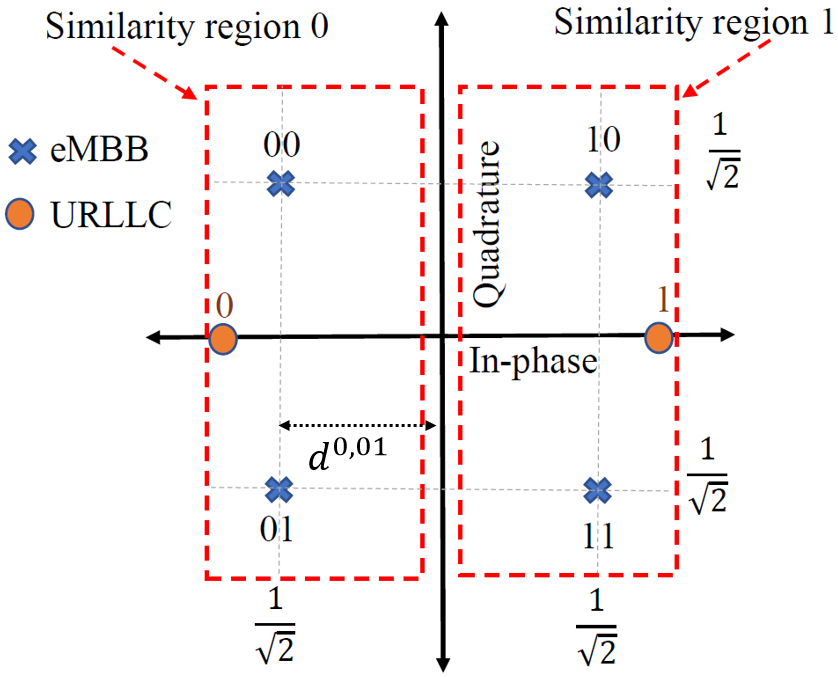}}
	\caption{Example for different modulation schemes and similarity regions.       $d^{0,01}$ is the minimum distance between the eMBB symbol 01 and the decision line of the URLLC symbol 0. }
	\label{Fig:Different_modulation_scheme}
\end{figure}
\begin{definition}
Let us consider two QAM schemes with modulation orders $m$ and $n$, respectively, and let us consider the diagram that has the superposition of their respective constellations. The similarity region of the two above modulation schemes is a region of the resulting constellation diagram that contains only one constellation point from the modulation that has the lowest order, i.e. $\min(m,n)$, and $\frac{\max(m,n)}{\min(m,n)}$ constellation points from the modulation that has the highest order, i.e., $\max(m,n)$, \textbf{which have the minimum Euclidean distance with the constellation point of the modulation that has the lowest order}. Based on this, there exist exactly $\min(m,n)$ similarity regions.
\end{definition}
As an illustration for Definition 1, let us consider the case when the eMBB traffic has a modulation order of $m=4$ and the URLLC traffic has a modulation order of $n=2$. The superposition of the constellations diagrams of the eMBB and URLLC modulations is shown in Fig.~\ref{Fig:Different_modulation_scheme}. The resulting diagram can be divided into $\min(m,n) = 2$ similarity regions, namely, similarity region 0 and similarity region 1, where each similarity region contains only one constellation point from the URLLC's modulation and $\frac{\max(m,n)}{\min(m,n)} = 2$ constellation points from the eMBB modulation that have the lowest Euclidean distance with the included constellation point of the URLLC modulation. Moreover, according to Definition 1, the {eMBB and URLLC symbols are divided into several sets and each set consists of several eMBB and URLLC symbols. The number of eMBB and URLLC symbols depends on the relation between the modulation orders of the eMBB and URLLC.} In practice, the modulation schemes of the URLLC and eMBB services may have the same order (i.e, $m=n$) or different ones ($m\neq n$). Accordingly, we classify the relationship between the eMBB and URLLC modulation orders into the following classes:
 	\begin{figure}[t] 
           		\centering
		\center{\includegraphics[width=1\columnwidth,draft=false]
		{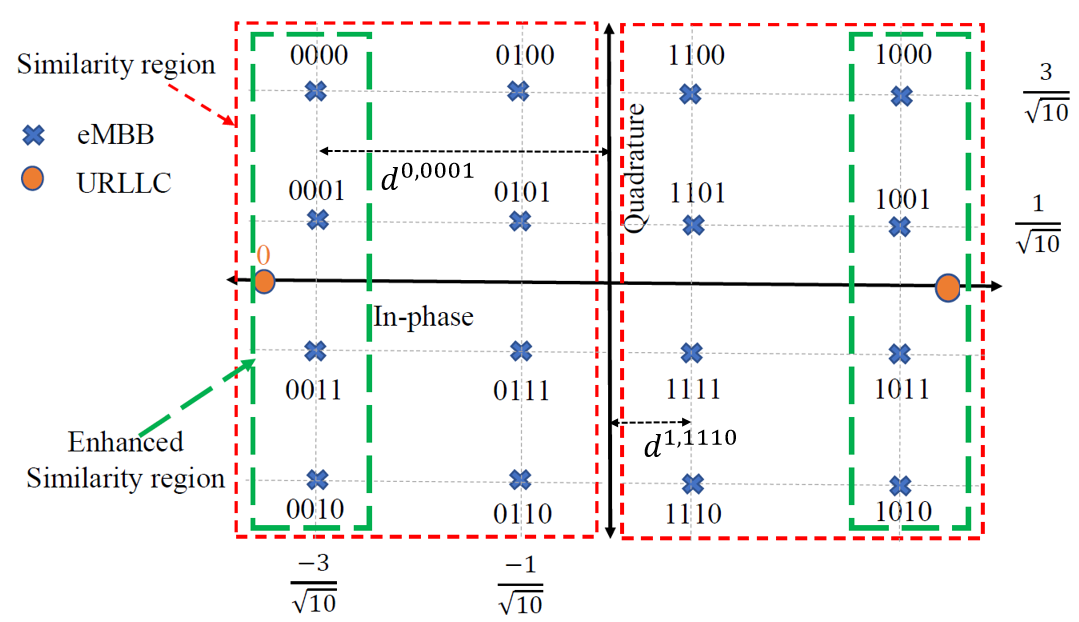}}
	\caption{Similarity region between eMBB traffic and URLLC load with 16-QAM (for eMBB) and BPSK (for URLLC). $d^{0,0001}$ is the minimum distance between the eMBB symbol 0001 and the decision line of the URLLC symbol 0.}
	\label{Fig:Similarity_region}
	\vspace{-.0 in}
\end{figure}
   \begin{itemize}
       \item Similar-Modulation-Order: In this case, the URLLC and eMBB have the same modulation order, i.e., $m=n$. Hence, each similarity region consists of one eMBB symbol and one URLLC symbol. This symbol is named as the \textit{Region-index-symbol}. 
              \item Lower-URLLC-Modulation-Order: In this case, the URLLC modulation order is lower than that of the eMBB, i.e., $m>n$. Accordingly, each similarity region consists of one URLLC symbol and $\frac{m}{n}$ eMBB symbols. Similarly, we can rename the URLLC symbol as the \textit{Region-index-symbol} and the eMBB symbols as \textit{mapping-symbols}.
       \item Higher-URLLC-Modulation-Order: In this case, the URLLC modulation order is higher than that of the eMBB, i.e., $m<n$. Accordingly, each similarity region consists of one eMBB symbol and $\frac{n}{m}$ URLLC symbols. We can rename the eMBB symbol as the \textit{Region-index-symbol} (that denotes the similarity-region) and the URLLC symbols as \textit{mapping-symbols}.
   \end{itemize}
   Practically, the Euclidean distance between the \textit{mapping-symbols} and the \textit{Region-index-symbol} varies according to their locations on the constellation. For clarity, as shown in Fig. \ref{Fig:Similarity_region}, assume the URLLC and eMBB traffic are modulated by the BPSK and 16-QAM, respectively. According to Definition 1, the constellation is divided into two similarity regions. The first region has URLLC symbol $0$ as the \textit{Region-index-symbol} and the second region has the URLLC symbols $1$ as the \textit{Region-index-symbol}. Without loss of generality, let the transmitted URLLC symbol be $0$. Then, the \textit{mapping-symbols} belonging to the same similarity region, i.e. \textbf{\{0000, \dots, 0111\}}, can be treated as \textbf{0} based on the BPSK maximum-likelihood receiver. On the other hand, the eMBB receiver receives the transmitted symbol correctly; hence the eMBB SER is not degraded. On the other hand, the SER of the URLLC becomes worse, as the symbol energy varies based on the eMBB constellation, which is 16-QAM in this example. We note that symbols \textbf{\{0000, 0010, 0011, 0001\}} have the lowest Euclidean distance with URLLC symbol $0$, hence they have lower SER at the URLLC receiver (See Fig. \ref{Fig:Similarity_region} for more elaboration.) In light of the above discussion, we define the symbol similarity as follows.
   \begin{definition}
   The similarity relation between the Region-index-symbol $s_{x}$ and the mapping-symbol $s_y$ in the same similarity region can be:
\begin{itemize}
    \item Absolute-similar: if $ P(\hat{s}\neq s_{x}|s_{x}\text{ was sent})-P(\hat{s}\neq s_{x}|s_y\text{ was sent})\ge 0$.
    \item Strongly-similar: if  $-\epsilon\le P(\hat{s}\neq s_{x}|s_{x}\text{ was sent})-P(\hat{s}\neq s_{x}|s_y\text{ was sent})<0$.
    \item Weakly-similar: if $P(\hat{s}\neq s_{x}|s_{x}\text{ was sent})-P(\hat{s}\neq s_{x}|s_y\text{ was sent})<-\epsilon$,
\end{itemize}
where $\epsilon\approx0$ depends on the target URLLC SER. Accordingly, we can call the set of symbols, which are absolute-similar and strongly-similar, as the enhanced similarity region.
   \end{definition}
\begin{definition}
The enhanced similarity region is a subset of the similarity region which includes the Region-index-symbol and mapping-symbols that satisfy $P(\hat{s}\neq s_{x}|s_y\text{ was sent}) -P(\hat{s}\neq s_{x}|s_{x}\text{ was sent})\le \epsilon$.
\end{definition}
\subsection{ URLLC Encoding at the BS}
According to the proposed puncturing strategy, it is preferred to puncture eMBB symbols such that the amount of symbol mismatch between the transmitted $\zeta-$symbols of the URLLC traffic and the punctured eMBB is minimized, i.e., smaller Hamming distance. Based on the URLLC-eMBB relationship, the encoding at the BS is illustrated as follows.
 \begin{itemize}
\item Similar-Modulation-Order: The BS encodes the URLLC traffic according to the desired modulation order $n$ while puncturing the eMBB symbol sequences that has maximum similarity (Absolute-similar), i.e., maximize the similar eMBB-URLLC OFDM symbols.  
\item Lower-URLLC-Modulation-Order: Similar to the Higher-URLLC-Modulation-Order case, the BS selects for puncturing the eMBB block that has a maximum number of absolute-similar, strongly-similar, and weakly-similar symbols. To accommodate the URLLC traffic, the BS can transmit either the encoded URLLC symbol or the ongoing eMBB symbol, as described below.
\item Higher-URLLC-Modulation-Order: the BS encodes the URLLC traffic according to the desired modulation order $n$ while puncturing the eMBB sequences that have a maximum similarity. In other words, the BS selects the eMBB symbol sequence that maximizes the number of absolute-similar, strong-similar, and weak-similar symbols. Compared to the Similar-Modulation-Order case, the impact of puncturing on the eMBB SER can not be eliminated.
\end{itemize}

 When the URLLC modulation order is lower than that of the eMBB, the BS can transmit the URLLC symbol or keep the ongoing eMBB symbol, as follows.
     \begin{itemize}
            \item URLLC mapper: The BS transmits the encoded URLLC symbols. Hence, the impact of puncturing on the eMBB resources can not be eliminated. To elaborate, let us consider the following example. If we have the following URLLC sequence \textbf{\{0, 1, 1, 0\}}, and the punctured eMBB sequence is \textbf{\{00, 11, 10, 01\}} (See Fig. \ref{Fig:Different_modulation_scheme}). According to the similarity region definition, these URLLC symbols are in the same similarity region of the punctured eMBB sequence. Assuming maximum-likelihood detection, we can roughly say 50\% of the punctured eMBB symbols will be correctly received, which translates to a high SER at the eMBB receiver (Fig. \ref{Fig:Different_modulation_scheme}.) 
            \item Similarity region mapper (SRM): To overcome the high SER of eMBB using the URLLC mapper, the BS transmits the eMBB symbol instead of the URLLC symbol, if they belong to the same similarity region, otherwise the URLLC symbol is transmitted. For example, assume that the modulation schemes of URLLC and eMBB are BPSK and 16-QAM, respectively, as shown in Fig. \ref{Fig:Similarity_region}. Also assume that the transmitted URLLC symbol is \textbf{0}. Then, any eMBB symbol belonging to the same similarity region, i.e. \textbf{\{0000, \dots, 0111\}}, will be received as \textbf{0} at the URLLC receiver with error probability less than \textbf{1}. On the other hand, the eMBB user receives the transmitted symbol correctly as the symbol is not affected by the puncturing process; hence the eMBB SER will improve. On the other hand, the SER of the URLLC becomes worse, since the eMBB symbols have different minimum distances from the URLLC decision boundary, which is 16-QAM in this example. 
            \item Enhanced Similarity region mapper (ESRM). To solve the high SER of the URLLC of the SRM, only the eMBB symbol belonging to the same enhanced similarity region, (that have better minimum distances from the URLLC decision boundary), i.e \textbf{\{0000, 0010, 0011, 0001\}}, are transmitted instead the URLLC symbol, otherwise the URLLC symbol is transmitted.(See Fig. \ref{Fig:Similarity_region} for more elaboration.)
    \end{itemize}
\begin{algorithm}[t!]
\SetAlgoLined
evaluate $\mathcal{P}_1 $ \;
       \uIf{$\mathcal{P}_1>\epsilon_u$}{
encode the URLLC packet using URLLC mapper\;
       }
  \Else{
    keep transmitting the eMBB symbols that satisfy the similarity conditions in Definition 2\;
  }
 \caption{Proposed SRM/ESRM mapper }
  \label{ALG:ENCODER}
\end{algorithm}    

Algorithm 1 illustrates an example for the mechanism of the SRM/ESRM. Let $\epsilon_u$ and $\mathcal{P}_1$ be the targeted SER of the URLLC traffic and the expected SER if the SR/ESRM is used, respectively. The algorithm starts by checking the activation condition $\mathcal{P}_1\le\epsilon_u$, if the activation condition is satisfied, the eMBB symbol is transmitted if they are  satisfying the similarity condition, otherwise the URLLC symbol will be sent.

 \subsection{ URLLC and eMBB Decoding}
	As mentioned above,  the eMBB receiver does not have the knowledge of the punctured symbols and it decodes the received sequence as if no puncturing took place. On the other hand, the URLLC receiver will perform the decoding process normally based on their modulation scheme. 
	For illustration, as shown in Fig. \ref{Fig:Different_modulation_scheme}, we assume that BPSK and QPSK are used to encode both the URLLC and eMBB traffic, respectively. Let the transmitted URLLC symbol be \textbf{\{0\}}. The eMBB decoder will translate the received symbol with probability($\approx50\%)$ as \textbf{\{00, 01\}}. Similarly, the URLLC receiver will decode the received symbol \textbf{\{00\}} as \textbf{\{0\}} with probability $(1-P_2(0.5\gamma_{\rm e}))$. {For the case when the ESRM is used, i.e., keep transmitting the similar eMBB symbols, the URLLC receiver will decode the received signal assuming the desired decoder, BPSK in this example. In summary, the decoding process at either the eMBB or the URLLC receivers is not changed.}


\section{Performance evaluation}\label{sec:evaluation}
\subsection{eMBB SER Analysis}
{Although both SER and bit error rate (BER) can be used to represent the impact of puncturing on the eMBB, the SER can represent the puncturing in the RE (symbol) level instead of the RB level which gives more sense about the proposed strategy.}
Hence, we use the SER, denoted by $P $, to measure the impact of the proposed puncturing strategy on the eMBB traffic. As  Without loss of generality, let $L_m=p_m L$ be the average number of eMBB symbols with modulation scheme $m$. Also, let the average number of eMBB symbols punctured due to the URLLC traffic with modulation order $n$ be $l_{n,m}$. According to the total probability theorem, the SER of the eMBB traffic under the effect of both the wireless channel and the presence of URLLC load can be expressed as
	
	\begin{equation}\label{eq:basic_ser}
	    \begin{split}
 P\left(\gamma_{\rm e},l\right)=\sum_{m=2}^{M}p_m\times\Bigg[   P_{m}(\gamma_{\rm e})\times \left(1-\frac{\sum_{n=2}^{N}l_{n,m}}{L_m}\right)\\+\sum_{n=2}^{N} P_{n,m}\left(\gamma_{\rm e},l_{n,m}\right) \times\frac{l_{n,m}}{L_m}\Bigg],\hspace{2 cm}	        
	    \end{split}
	\end{equation}
	where $P_{m}(\gamma_{\rm e})$ is the SER of the eMBB traffic with modulation order $m$ due to the channel error only, and $P_{n,m}(\gamma_{\rm e},l_{n,m})$ is the SER due to the channel and URLLC traffic with modulation order $n$. To quantify the actual effect on the eMBB traffic, we start with the following definition.

 \textbf{Definition 4:} \textit{ Consider an eMBB and URLLC traffic with modulation orders $m$ and $n$, respectively. The average effectively punctured symbols}, $\mathcal{L}_{n,m}$, of eMBB traffic is a portion of the punctured eMBB symbols, $l_{n,m}$, in which the transmitted URLLC symbol, $s_u$, has a different similarity region from that of the punctured eMBB symbol, $s_e$, and its range is $0\le \mathcal{L}_{n,m} \le l_{n,m}$.
\begin{figure}[!t]
\vspace{-.3 in}
	\centering
	\center{\includegraphics[width=1\columnwidth,draft=false]
		{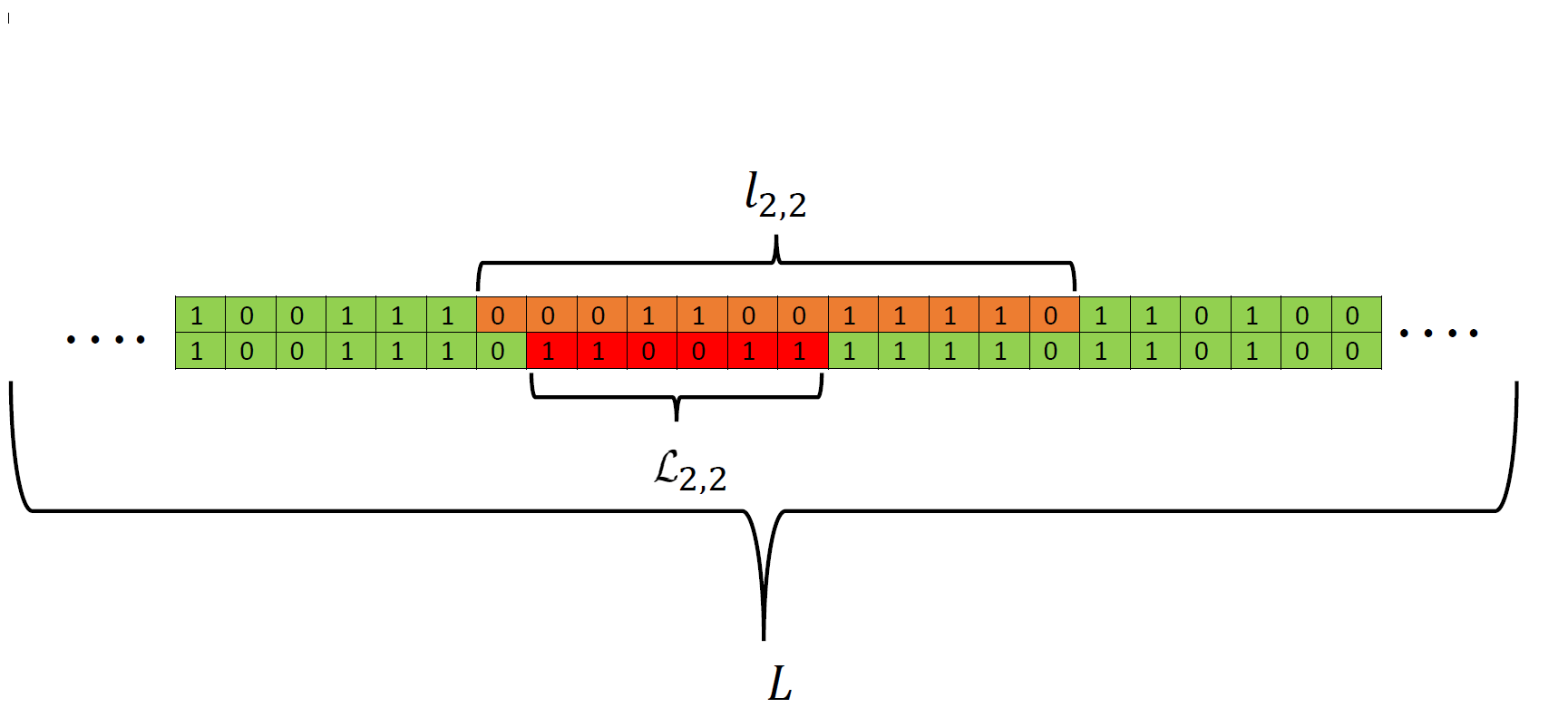}}
	\caption{Effectively punctured eMBB symbols.}
	\label{Fig:leff}
	\vspace{-.1 in}
\end{figure} 
 Fig. \ref{Fig:leff} illustrates the relation between $\mathcal{L}_{n,m}$ and $l_{n,m}$ for the case of similar modulation schemes, BPSK in this example. It shows that, due to the similarity region between the punctured eMBB and URLLC symbols, some of the punctured eMBB symbols are not affected by the puncturing process. In general, for eMBB and URLLC traffic with modulation orders $m$ and $n$, respectively, the expected number of the \textit{effectively eMBB punctured} symbols is $\mathcal{L}_{n,m}<=l_{n,m}$. Accordingly, (\ref{eq:basic_ser}) can be written as follows.
 \begin{equation}\label{eq:ser_l_eff}
\begin{split}
 P\left( \gamma_{\rm e},l \right)= \sum_{m=2}^{M} p_m\times\Bigg[  P_{m}\left(\gamma_{\rm e}\right)\times \left(1-\frac{\sum_{n=2}^{N}l_{n,m}}{L_m}\right)
    +\quad\\ \sum_{n=2}^{N}\left( {P}_{n,m}(\gamma_{\rm e},\overline{\mathcal{L}}_{n,m}) \times\frac{ \overline{\mathcal{L}}_{n,m}}{L_m}
     +{P}_{n,m}\left(\gamma_{\rm e},\mathcal{L}_{n,m}\right) \times\frac{\mathcal{L}_{n,m}}{L_m}\right)\Bigg],
 \end{split}
  \end{equation}
	where ${P}_{n,m}(\gamma_{\rm e},{\overline{\mathcal{L}}}_{n,m})$ and ${P}_{n,m}\left(\gamma_{\rm e},\mathcal{L}_{n,m}\right)$ are the error probabilities (which will be derived in the sequel) of the eMBB traffic in $\overline{\mathcal{L}}_{n,m}$ and ${\mathcal{L}}_{n,m}$, respectively. Where $\overline{\mathcal{L}}_{n,m}=l_{n,m}-{\mathcal{L}}_{n,m}$.

	Equation (\ref{eq:ser_l_eff}) shows the SER of the eMBB traffic under the impact of the wireless channel and the presence of URLLC load. The first term represents the average error probability for the fraction of eMBB sequence impacted by the channel errors only. The second term represents the average error probability of the punctured eMBB symbols that have the same similarity region to the URLLC symbols. The third term is the average error probability of the effectively punctured eMBB symbols. 

{
}
\subsection{Puncturing Parameters Evaluation}
In this section, we analyze the puncturing parameters ${\mathcal{L}}_{n,m},$ ${P}_{n,m}(\gamma_{\rm e},{\overline{\mathcal{L}}}_{n,m})$ and ${P}_{n,m}\left(\gamma_{\rm e},\mathcal{L}_{n,m}\right)$  for the proposed puncturing strategy. The puncturing parameters depend on the modulation schemes of eMBB and URLLC, and how the URLLC traffic is distributed, i.e., $l_{n,m}$. Without loss of generality, the average punctured eMBB symbols (the average URLLC load) is assumed to be known as it depends on the arrival rate, $\lambda$, of the URLLC traffic. Also, the channel error, $P_m(\gamma_{\rm e})$, depends on the channel and the modulation scheme. For example, the SER of eMBB under additive white Gaussian noise (AWGN) and/or Rayleigh fading with SNR per symbol ${\gamma}_{e}$ are \cite{6620186,goldsmith2005wireless}:
	\begin{equation}\label{eq:channel_ser}
 P_m(\gamma_{\rm e}) \approx
 	 \begin{cases}
 	 4a\, Q(\sqrt{\frac{3 {\gamma}_{e}}{m-1}}),  \hspace{1.18 in}\text{AWGN}, \\ 
 	 2a\, (1-b)-a^2(1-\frac{4b}{\pi}\tan^{-1}(\frac{1}{b})), \text{Rayleigh,}
 	 \end{cases}
 	 \end{equation}
where $a=(1-\frac{1}{\sqrt{m}})$ and $b=\sqrt{\frac{3 {\gamma_{\rm e}}}{2(m-1)+3 {\gamma_{\rm e}}}}$.\\ 
\subsubsection{Average Effectively Punctured Symbols} Without loss of generality, assume the distribution of the symbol similarity between a URLLC block and a punctured eMBB sequence with modulation orders $n\text{ and }m$ follows the Binomial distribution $B~(\zeta,\eta_{n,m})$. Therefore, the average similarity between the two traffic blocks can be defined as $U_{n,m,\zeta}\triangleq\zeta \times \eta_{n,m}$, where $\eta_{n,m}$ (probability that any two symbols are similar) is
\begin{equation}\label{eq:pijpinew}
 \eta_{n,m}=\sum_{j=0}^{n-1}\sum_{i=0}^{m-1}p_j\,  p_i,\forall i,j\in \Omega ,
		\end{equation}
where $p_j$ and $p_i$ are the probabilities of sending symbol $j,i$ of the eMBB and URLLC traffic, respectively. Accordingly, we can represent the average effectively punctured eMBB symbols, in terms of the average URLLC load in the following definition.   

\textbf{Definition 5:} Consider a URLLC and eMBB traffic with modulation orders $n \text{ and } m$, respectively. Also, let the eMBB block length and the average length of punctured eMBB symbols by the URLLC traffic be $L_m$ and $l_{n,m}$, respectively. The average effectively punctured eMBB symbols (i.e., modified) is given by:
\begin{equation}\label{eq:lemma1}
		 {\mathcal{L}}_{n,m}=\left(1-\frac{\mathcal{U}_{n,m}(\zeta,l_{n,m},L_m)}{\zeta}\right)\, l_{n,m} \,(\text{symbols}),\vspace{-.01 in}
\end{equation}
where $\mathcal{U}_{n,m}(\zeta,l_{n,m},L_m)$ is the expected number of eMBB symbols which have the same similarity region with the transmitted $\zeta-$symbols of URLLC block. The term $\frac{\mathcal{U}_{n,m}(\zeta,l_{n,m},L_m)}{\zeta}$ represents the ratio (percentage) of similarity between both services. Then, $\left(1-\frac{\mathcal{U}_{n,m}(\zeta,l_{n,m},L_m)}{\zeta}\right)\, l_{n,m}$ is the actual punctured eMBB symbols.
 
The definition in (5) gives an expression for the average length of the effectively punctured (modified) eMBB symbols. However, it does not quantify the average similarity, $\mathcal{U}_{n,m}(\zeta,l_{n,m},L_m)$, between the URLLC and eMBB sequences. In fact, $\mathcal{U}_{n,m}(\zeta,l_{n,m},L_m)$ depends on the URLLC block size and the eMBB search space of size ${N}<L$. Lemma \ref{lemma2} gives an approximated value for $\mathcal{U}_{n,m}(\zeta,l_{n,m},L_m)$.

\begin{lemma}\label{lemma2} Let $L_m$ denote the average eMBB traffic, and let $l_{n,m}$ be the average number of punctured eMBB symbols. Assume that the URLLC traffic is divided into blocks with $\zeta$-symbols each. An upper bound on the expected similarity between the URLLC and eMBB traffic is given by
\begin{equation}\label{eq:Uij}
	 \mathcal{U}_{n,m}(\zeta,l_{n,m},L_m)=\frac{1}{\lceil\frac{l_{n,m}}{\zeta}\rceil}\sum_1^{\lceil\frac{l_{n,m}}{\zeta}\rceil}\sum\nolimits_{k=0}^{\zeta-1}\left[1-\{F(k)\}^{L_{m}-\zeta}\right],
\end{equation} 
where, $F(k)=\sum_{{j}=0}^{k}\left(\begin{array}{l}{\zeta} \\{j}\end{array}\right) \left( \eta_{n,m}\right)_j(1-\eta_{n,m})^{\zeta-j}$.
\end{lemma}
 \begin{proof}
  See Appendix A for the proof.
 \end{proof}
For large $\zeta$, $\mathcal{U}_{n,m}(\zeta,l_{n,m},L_m)$ reduces to $\zeta \eta_{n,m}$. Hence, we can further reduce the block size, i.e., $\zeta$, and consequently the average similarity will increase. Practically, decreasing the size of $\zeta$ will increase the signalling overhead. Therefore, a proper selection of $\zeta$ is important. Indeed, we examine the effect of different values of $\zeta$ on the performance of the proposed scheme in Section V.\\ 
\subsubsection{SER of the Effectively Punctured Symbols}
 The SER of the Effectively Punctured Symbols, ${P}_{n,m}\left(\gamma_{\rm e},\mathcal{L}_{n,m}\right)$, strictly depends on the relation between the modulation schemes of both URLLC and punctured eMBB traffic. $\mathcal{L}_{n,m}$ is the average number of eMBB symbols that are wrongly transmitted. In other words, the transmitted symbols belong to another region. ${P}_{n,m}\left(\gamma_{\rm e},\mathcal{L}_{n,m}\right)$ is then expressed as
\begin{equation}
 {P}_{n,m}\left(\gamma_{\rm e},\mathcal{L}_{n,m}\right)=\sum_{s_j\in\Omega}\sum_{s_i\notin\Omega} p(s_j|s_i \text{ sent}) P_{err}(s_j|s_i \text{ sent}),
\end{equation}
where $s_j$ and $ s_i$ are the effectively punctured eMBB symbol and the transmitted URLLC symbol, respectively. We can upper bound ${P}_{n,m}\left(\gamma_{\rm e},\mathcal{L}_{n,m}\right)$ with the closed form expression 
 \begin{equation}
     {P}_{n,m}\left(\gamma_{\rm e},\mathcal{L}_{n,m}\right)\le 1-P_n\left(\gamma_{\rm e}\right)\times\left(\frac{1}{m-1}\right)\approx 1,
 \end{equation}
  where $P_n\left(\gamma_{\rm e}\right)$ is the probability of error under the URLLC modulation condition. For example, let BPSK be the modulation order used by both URLLC and eMBB traffic and $P_{2}(\gamma_{\rm e})=10^{-2}$, then  $P_{2,2}(\gamma_{\rm e},\mathcal{L}_{n,m})=1-10^{-2}\approx1$.\\ 
\subsubsection{SER of the non-Effectively Punctured Symbols}
The SER of the non-Effectively punctured symbols, ${P}_{n,m}\left(\gamma_{\rm e},\overline{\mathcal{L}}_{n,m}\right)$ is summarized as follows.
\begin{itemize}
    \item Similar-Modulation-Order: The modulation schemes of URLLC and eMBB are similar. Accordingly, non-effectively punctured symbols are similar to the punctured symbol, hence ${P}_{n,m}\left(\gamma_{\rm e},\overline{\mathcal{L}}_{n,m}\right)$ is expressed as:
    \begin{equation}
      {P}_{n,m}\left(\gamma_{\rm e},\overline{\mathcal{L}}_{n,m}\right)={P}_{m}\left(\gamma_{\rm e}\right)
    \end{equation}
    \item Lower-URLLC-Modulation-Order: similar to Similar-Modulation-Order, in this case, the non-effectively punctured eMBB symbols are only affected by the channel conditions. Hence, ${P}_{n,m}\left(\gamma_{\rm e},\overline{\mathcal{L}}_{n,m}\right)$ is expressed as follows:
    \begin{equation}
      {P}_{n,m}\left(\gamma_{\rm e},\overline{\mathcal{L}}_{n,m}\right)={P}_{m}\left(\gamma_{\rm e}\right)
    \end{equation}
    \item Higher-URLLC-Modulation-Order: As the energy of the transmitted URLLC symbols varies, an exact expression for the URLLC SER is not easy to obtain. Hence, we obtain an upper bound for  ${P}_{n,m}\left(\gamma_{\rm e},\overline{\mathcal{L}}_{n,m}\right)$ based on the minimum distance, ${d}^{i,j}$, the transmitted URLLC symbol and the decision boundary of the eMBB symbol. Accordingly, ${P}_{n,m}\left(\gamma_{\rm e},\overline{\mathcal{L}}_{n,m}\right)$ is expressed as:
    \begin{equation}
 {P}_{n,m}\left(\gamma_{\rm e},\overline{\mathcal{L}}_{n,m}\right)=\sum_{s_j\in\Omega}\sum_{s_i\in\Omega} p(s_j|s_i \text{ sent}) P_{err}(s_j|s_i \text{ sent}),
\end{equation}
where $P_{err}(s_j|s_i \text{ sent})$ is expressed as 
    \begin{equation}
 P_{err}(s_j|s_i \text{ sent})=P_j(\gamma_{\rm e}{{d^{i,j}}^2})
\end{equation}
\end{itemize} 	
\subsection{URLLC SER Analysis}
The SER of the URLLC traffic is only affected when the URLLC traffic is modulated using the SR/ESRM. This is because the the energy of the transmitted symbols are different than the actual URLLC symbols. In other words, the energy of the non-effectively punctured eMBB symbols is varies. Accordingly, an exact expression for the URLLC SER is not easy to obtain. Hence, we drive an upper bound expression for the URLLC SER based on the minimum distance of the transmitted symbol and decision boundary, as 
\begin{equation}\label{eq:ser_urllc}
\begin{split}
   \mathcal{P}_{n,m}\left(\gamma_{\rm u}\right)=\left(1-\sum_{s_j}\sum_{s_i} p\left(s_i|s_j \text{ sent}\right)\right)\mathcal{P}_{n}\left(\gamma_{\rm u}\right)  \\
   +\sum_{s_j}\sum_{s_i} p\left(s_i|s_j \text{ sent}\right) \mathcal{P}_{n}\left(\gamma_{\rm u} {d^{i,j}}^2\right),
\end{split}
\end{equation}
Equation (\ref{eq:ser_urllc}) shows the SER of the URLLC traffic when the SR/ESRM is used. The first term represents the average error probability for the fraction of URLLC sequence impacted by the channel errors only. The second term represents the average error probability of the URLLC symbols that have the same similarity region to the eMBB symbols (encoded by the SRM). {In fact, the SER loss is equivalent to a power loss, $W_{dB}$ which has the following expression:
\begin{equation}\label{eq_urlls_loss_db}
    W_{dB}\approx10 \sum_{s_j}\sum_{s_i} p\left(s_i|s_j \text{ sent}\right) \log_{10}\left( {d^{i,j}}^2/{d^{i}}^2\right).
\end{equation}}
The expression in (\ref{eq_urlls_loss_db}) evaluates the average URLLC loss in dB. The term $\log_{10}\left( {d^{i,j}}^2/{d^{i}}^2\right)$ is the power loss for each URLLC symbol in terms of the ratio between the distance of the URLLC and the transmitted eMBB symbols form the decision boundary.


\subsection{SER Scaling}
In light of the above discussion, we can observe that the eMBB SER is a function the SNR and the average similarity of the punctured eMBB symbols. When the SNR increases, the SER improves and it is asymptotically equal to the puncturing errors. Then, the eMBB SER based on (\ref{eq:ser_l_eff}) is approximated as

	\begin{equation}\label{eq:ser_scale}
 	  P\left(l\right) \approx \sum_{m=2}^{M} p^m
 \frac{	 \sum_{n=2}^{N}{\mathcal{L}}_{n,m}}{L_m}.
 	 \end{equation}
On the other hand, as the similarity increases ($L$ increases or $\zeta$ decreases), the eMBB SER reduces to only channel errors, which can be approximated as
\begin{equation}\label{eq:scaling}
   \begin{split} 
 P\left(\gamma_{\rm e},l\right)\approx\sum_{m=2}^{M}p_m \times \Bigg[  P_{m}\left(\gamma_{\rm e}\right)\times \left(1-\frac{\sum_{n}^{N}\overline{\mathcal{L}}_{n,m}}{L_m}\right)\\+\sum_{n}^{N}{P}_{n,m}\left(\gamma_{\rm e},\overline{\mathcal{L}}_{n,m}\right)\frac{\overline{\mathcal{L}}_{n,m}}{L_m}\Bigg].
    \end{split}
\end{equation}
We observe that the eMBB performance strictly depends on the SNR and the average similarity. As the SNR increases, the eMBB SER becomes dominated by the puncturing errors, while increasing the similarity will reduce the SER to errors due to the channel condition. 

\subsection{eMBB Loss Function}
The function that represents the eMBB loss associated to the puncturing schemes can be either a linear, convex or a threshold function \cite{joint}. The linear loss function, i.e. $h(x)=\alpha\, x$, has been widely used to study the impact of the superposition/puncturing scheme \cite{joint,8643428}. The expected loss of an eMBB traffic, based on the linear function, is the ratio between the punctured eMBB symbols to total eMBB symbols:  
 \begin{equation}\label{eq:cost_eMBB}
   E[h\left(l_m\right)]= \frac{\sum_n^N l_{n,m}}{L_m},
 \end{equation}
 where $l_m=\sum_n l_{n,m}$. The loss function in (\ref{eq:cost_eMBB}) is widely coupled with the eMBB rate \cite{8476595,8638930,8932425,8663990,joint,8643428,8746407}. Taking into consideration the  effectively punctured symbols and similar symbols, and using results in (2) and (5), the expected loss in (\ref{eq:cost_eMBB})  can be generalized as follows:
  \begin{equation}\label{eq:cost_Proposed}
   \begin{split}
      E[h\left(l_m\right)]=\frac{1}{L_m}\sum_n^N\bigg( { {P}_{n,m}\left(\gamma_{\rm e},\overline{\mathcal{L}}_{n,m}\right)\overline{\mathcal{L}}_{n,m}}\\
      +{ {P}_{n,m}\left(\gamma_{\rm e},{\mathcal{L}}_{n,m}\right){\mathcal{L}}_{n,m}}\bigg).
   \end{split}
 \end{equation}
For clarity, assume a retransmission-based puncturing is adopted. Then, all punctured eMBB symbols are lost, i.e.,  $E[h\left(l_m\right)]=\frac{l_m}{L_m} .$

\subsection{Proposed Search Algorithm} \label{gh}

{
\begin{algorithm}[t!]
\SetAlgoLined
$c_{k} \leftarrow 0 \, \forall k \in[1, {K}] $\;
\textbf{Step 1:} Similarity weight calculation\;
 \For{ $k=1\rightarrow {K}$}
 {
  \For{$t=1\rightarrow \zeta$}
  { 
  count symbols in the same similarity region\;
  \If{${s}_e^t\,\&\,{s}_u^t\in \Omega$ }
    {$\mathbf{c}_k\leftarrow\mathbf{c}_k+1$\; 
   }
   }
 }
 \textbf{Step 2:} eMBB block selection\;
 $k^{*} \leftarrow \underset{k \in[1, {K}]}{\arg \max } \,c_{k}$\;
 \caption{Proposed search Algorithm}
  \label{alg:1}
\end{algorithm}
}

The latency constraint is a critical factor to maintain QoS of the URLLC service. Therefore, we present a fast search algorithm, of time complexity $O({K})$, that exploits the similarity between the URLLC block over multiple eMBB traffic sequences with different modulation schemes. Initially, the algorithm associates a counter $c_k$ to each eMBB in ${K}$. The subset, ${K}$, of the possible eMBB blocks for puncturing depends on the latency requirement of the URLLC traffic. Specifically,  ${K}$ should be small for the URLLC traffic with strict latency constraint, and a relatively larger ${K}$ otherwise. Moreover, the URLLC  traffic with weak latency constraint implies that several mini-slots (URLLC-slot) is allowed for allocating the URLLC block.\\
\indent As shown in Algorithm \ref{alg:1}, the proposed algorithm has two steps: the first step counts the similar symbols between the eMBB blocks and the URLLC block; the second step selects the suitable eMBB block for puncturing. The first step contains two loops, inner and outer loop, which describe the number of eMBB blocks for possibly punctured and the number of URLLC symbols per block. As illustrated in Algorithm 1, the similarity region of the eMBB symbol, $s_e^t $, is compared with the similarity region of the URLLC symbol, $s_u^t$. Accordingly, the counter $c_k$ is incremented by one when both symbols are in the same similarity region, otherwise it remains not incremented. In the second part, the BS selects the eMBB block that has maximum similarity with the URLLC block. In other words, the punctured eMBB block $k^{*}$  should have a maximum count of similar symbols with the URLLC block.\\
\indent To mitigate adverse impact on eMBB traffic, the URLLC load is segmented into smaller blocks, i.e., 1 RB \cite{8746407}. Hence, Algorithm 3 guarantees that all URLLC segments are delivered in the correct order, i.e., the URLLC symbols sequence is correct. The algorithm initially divides the search space, ${K}$, into ordered and equal subsets, ${K}_1<{K}_2<{K}_3...<{K}_Z$, where $Z$ is the number of URLLC segments.  Then, Algorithm 2 is applied to each segment on the corresponding subset. To enhance the algorithm, the remaining eMBB blocks of subset, ${K}_k$, which satisfy the inequality $k^{*}_z< k_z<{K}_k$, are merged with the next subset, ${K}_{k+1}$, using the $merge()$ function.\\
\indent {\textbf{Search algorithm time complexity :} It can be easily shown that the proposed algorithm has a time complexity of $O({K})$ which make it an efficient and practical solution.  For clarity, the search algorithm (Algorithm 2) consists of two steps, i.e., \textbf{Step 1} and \textbf{Step 2}. \textbf{Step 1} consists of one outer loop and one inner loop of ${K}$ and $\zeta$ iterations, respectively. Hence, for a URLLC block with fixed number of symbols $\zeta$, \textbf{Step 1} has a time complexity of $O( {K})$. \textbf{Step 2} aims to select the maximum counter over ${K}$ elements. In general, \textbf{Step 2} performs ${K}$ comparison which means \textbf{Step 2} also has a time complexity of $O({K})$. As a result, the proposed algorithm (Algorithm 1) has a low time complexity of $O({K})$.} 

{
\begin{algorithm}[t!]
\SetAlgoLined
${k}^*_{z} \leftarrow 0 \, \forall z \in[1, Z] $\;
 \For{ $z=1\rightarrow Z$}
 {
 $k^{*}_z\leftarrow Algorithm\,1({K}_z)  $\;
  $N_{k+1}\leftarrow merge(k^{*}_z+1\rightarrow {K}_{z} , {K}_{z+1})  $\;
 }
 \label{alg:2}
 \caption{Search segmentation Algorithm.}
\end{algorithm}
}

\section{Simulation Results}\label{sec:results}

\subsection{Simulation Setup}
\indent In this section, we carry out various simulations to evaluate the performance of the proposed puncturing strategy. We consider a wireless network which consists of one BS and eMBB and URLLC traffic. We assume that the eMBB traffic belongs to 10 eMBB users and the URLLC packet arrival follows the Poisson distribution with arrival rate, $\lambda$ and each packet size is $\zeta=96$ bits. We also assume that the BS has $L=1200$ downlink frequency resources (RE). We assume the channel between the BS and the eMBB and URLLC receivers is Rayleigh fading, wherein the eMBB channel gain remains constant for two time slots (14-sTTI) and one eMBB block is transmitted within this period. {Moreover, we assume the served users are located in different distance from the BS, hence the path loss is considered with a path loss exponent equals $3$.} The noise at the receiver is assumed to be complex AWGN $\mathcal{CN}(N_0,0)$, where $N_0=10^{-9}$ is the noise power. The BS can use BPSK, 4-QAM, 16-QAM or 64-QAM to modulate the eMBB traffic. Particularly, the BS adopts the modulation order $m\in\{4, 16, 64\}$ such that the channel SER is less than or equal $0.01$, otherwise BPSK is adopted. Moreover, we assume that the CSI of the URLLC  traffic is not available at the BS. Hence, the URLLC traffic is modulated using only BPSK  to achieve the maximum reliability.\\
\indent The performance of the similarity puncturing strategy is evaluated for different transmitting power and arrival rate, i.e., $\lambda=7$ and $\lambda=3.5$ packets per millisecond (p/msec). {we assume that the eMBB users are unaware of the punctured resources, so we consider the code-based puncturing proposed in \cite{8403963,pedersen2017punctured} as a baseline algorithm. Generally, when the eMBB receiver is unaware of the punctured resources of the transmission, the received signal is decoded as useful signal. Resource proportional (RP) placement is used to allocate the URLLC traffic as it gives the optimal solution for the linear loss model \cite{joint}.}\\
\indent{ In the following, we start by showing the advantage of the proposed strategy on the performance of the eMBB traffic in terms of the spectral efficiency, SER and reliability. Second, we investigate the performance of the proposed strategy on the URLLC traffic by considering the URLLC SER and reliability. Finally, we evaluate the computational time of the proposed puncturing strategy. }
    \begin{figure}[!t] 
	\centering
	\center{\includegraphics[width=1\columnwidth,draft=false]
		{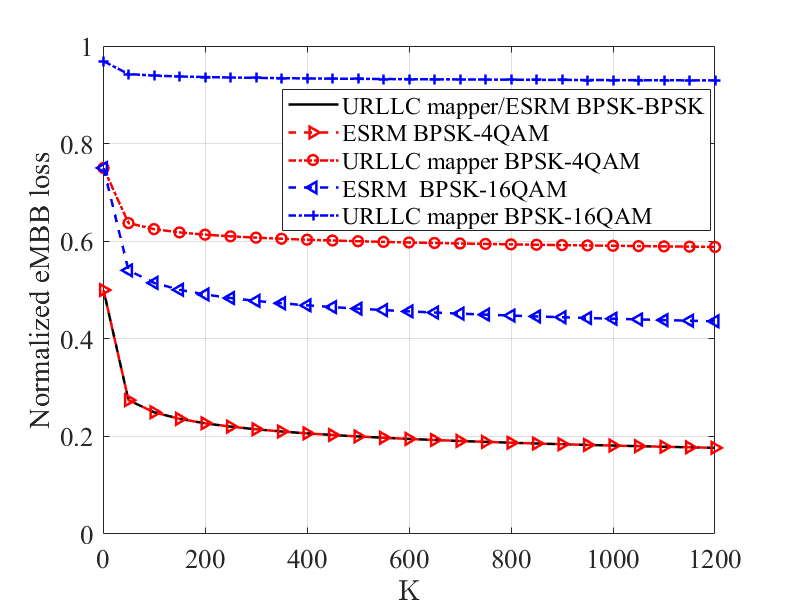}}
	\caption{Average eMBB loss relative to the punctured symbols. Adopted $\zeta=24~ (1~RB)$ and $\gamma_{\rm e}=40~dB$ }
	\label{Fig:loss2}
	\vspace{-.0 in}
\end{figure}
\subsection{eMBB Traffic Performance}
\subsubsection{Spectral Efficiency} In (\ref{eq:cost_Proposed}), we expressed the loss function of the eMBB traffic in terms of the SER of the punctured symbols. To measure the efficiency of the proposed strategy and select the optimal ${K}$, we evaluate the average eMBB loss (contaminated eMBB symbols) for both the URLLC mapper and the ESRM while varying the size of the search space ${K}$ (see Fig. \ref{Fig:loss2}). The results show that the percentages of the contaminated (lost) eMBB symbols for the ESRM are 18\% and 44\% compared to 59\% and 93\% for the URLLC mapper for BPSK-4QAM and BPSK-16QAM, respectively. { This enhancement of the ESRM results from transmitting the punctured eMBB symbols that fall in the same similarity region of transmitted URLLC symbols ( i.e., punctures the eMBB symbols which has different similarity region).} Moreover, the results show that the ESRM for (BPSK-4QAM) achieves the   same loss of (BPSK-BPSK). This because the probability of similarity in the similarity region is the same, i.e., 0.5. {However, the eMBB SER is enhanced for BPSK-4QAM using ESRM, this comes at the expense of the URLLC SER which decreases by 2.5 dB according to equation in  (\ref{eq_urlls_loss_db}).} The results show that the proposed scheme performs better than the code-based puncturing, i.e., ${K}=1$ for URLLC mapper. When ${K}$ is small, it indicates that a small number of eMBB blocks can be punctured, due to the eMBB QoS requirement.  For example, the point when ${K}=300$, the search algorithm scans only $300$ eMBB blocks (possibilities) out of ${K}=1200$ to allocate the URLLC traffic. In other words, the QoS requirements limits ${K}$ (number of eMBB blocks the URLLC is compared with): which means that the smaller ${K}$, the more the eMBB traffic with strict QoS requirements. \\
\begin{figure}[tb!]
\centering
\begin{subfigure}[b]{1\columnwidth}
\centering
\includegraphics[width=1\columnwidth,draft=false]{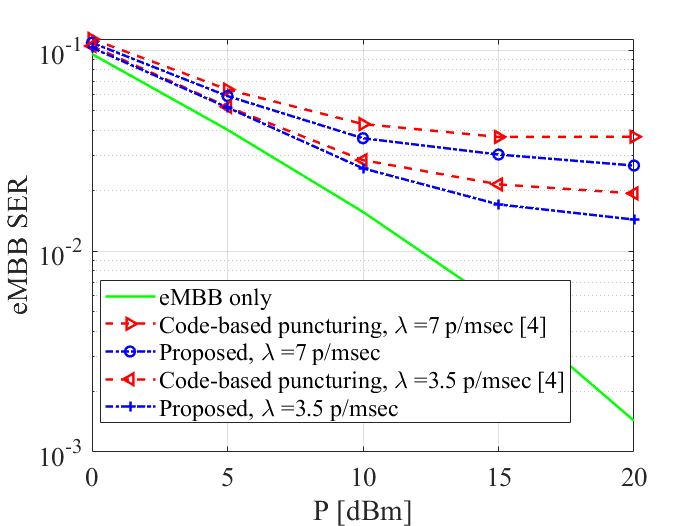}
\caption{eMBB SER vs transmission power in dBm for different URLLC arrival rate.}
\label{fig:ser_embb}
\end{subfigure}\\
\begin{subfigure}[b]{1\columnwidth}
\centering
\includegraphics[width=1\columnwidth,draft=false]{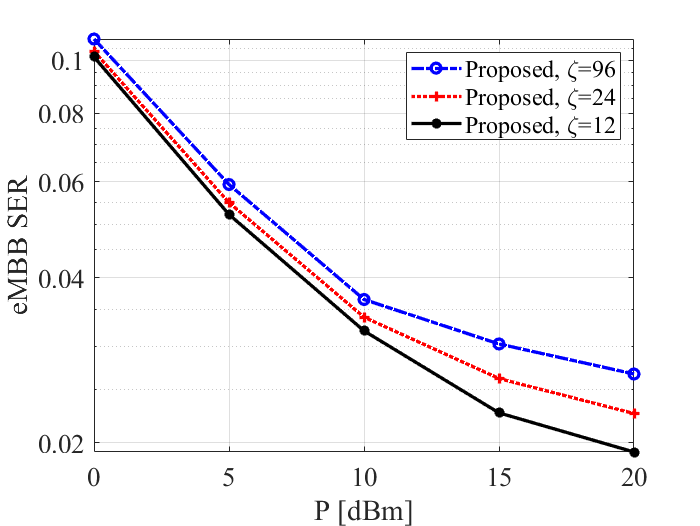}
\caption{eMBB SER vs transmission power in dBm for different URLLC block (segment) size. $\lambda=7~p/msec$ }
\label{Fig:SEG}
\end{subfigure}
\caption{ eMBB SER of the proposed puncturing strategy vs transmission power in dBm.}
\label{fig:ser_per}
\end{figure}
\subsubsection{eMBB SER} Fig. \ref{fig:ser_per} illustrates  the performance of the proposed puncturing scheme in terms of the SER of eMBB. We make the following observations from the results.
\begin{itemize}
    \item The proposed algorithm achieves better SER for the eMBB traffic compared to the code-based baseline. In other words, the proposed scheme can achieve the target eMBB SER with lower transmission power. Particularly, the achieved gain increases with the SNR (transmission power) (it reaches 10 dB at high SNR), as shown in Fig. \ref{fig:ser_embb}. Also, the figure shows that the gain of our proposed method is negligible at low transmission power (low SNR) since the channel errors are the dominant here, however the gain improves as the SNR increases since the error at high SNR is dominated by puncturing. Note also that the gain saturates at high SNR according to (\ref{eq:ser_scale}), with no further improvement as the BS allocates more transmit power for the eMBB symbols (puncturing errors dominates).
    \item It is intuitive that the SER of the eMBB traffic deteriorates as the URLLC load increases, as illustrated in Fig. \ref{fig:ser_embb}. For instance, the eMBB SER saturates at $0.02$ and $0.04$ at both $\lambda=3.5$ and $\lambda=7$, respectively. 
    This increase in SER is due to that as the URLLC load increases, the more eMBB resources are punctured, which leads to more SER. 
    \item  Fig. \ref{Fig:SEG} illustrates the eMBB SER for different URLLC block segment, $\zeta$. For instance, eMBB SER is enhanced when the URLLC block size $\zeta$ becomes smaller, as illustrated in Fig \ref{Fig:SEG}. Reducing $\zeta$ increases the probability of similarity, which decreases the effectively punctured symbols. In this case, the trade-off between reducing $\zeta$ and the overhead should be optimized to achieve better spectral efficiency. In other words, as $\zeta$ decreases, the overhead increases.
    \end{itemize}
\begin{figure}[tb!]
\centering
\begin{subfigure}[b]{1\columnwidth}
\centering
\includegraphics[width=1\columnwidth,draft=false]{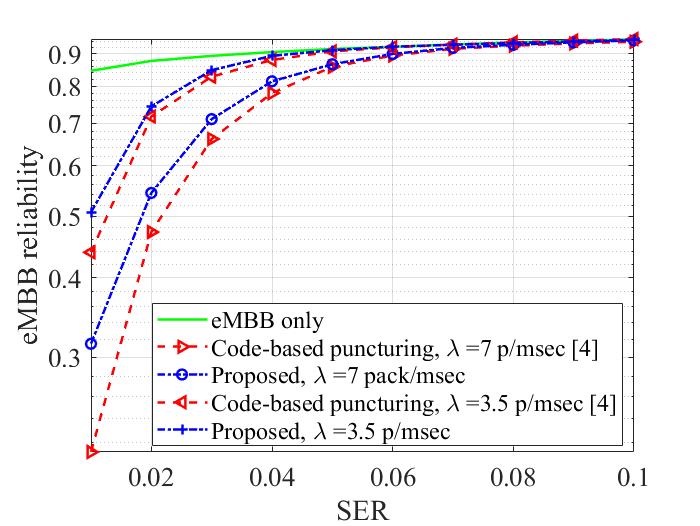}
\caption{eMBB reliability vs the SER.}
\label{fig:embb_rel}
\end{subfigure}\\
\begin{subfigure}[b]{1\columnwidth}
\centering
\includegraphics[width=1\columnwidth,draft=false]{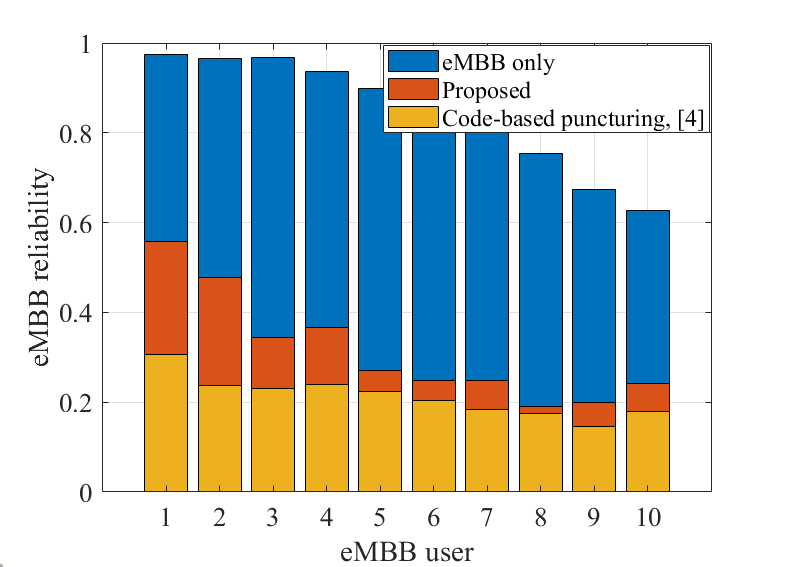}
\caption{ eMBB users reliability. Adopted SER=.01 and $\lambda=7 p/msec$}
\label{fig:users_rel}
\end{subfigure}
\caption{ eMBB traffic reliability of the proposed algorithm.  $P=10~dBm$ }
\label{fig:performanceanalysis_SNR}
\end{figure}
\subsubsection{eMBB Reliability}
In this section, we evaluate the reliability of the eMBB traffic as a function of the achieved SER. In general, we define the reliability of both URLLC and eMBB traffic
\begin{equation}
\text{reliability}=\frac{\text{Number of blocks satisfying the targeted SER}}{\text{Total transmitted blocks}}.
\end{equation}
 Fig.~\ref{fig:embb_rel} presents the eMBB reliability while varying the targeted BER, for different URLLC arrival rates. The figure shows that the proposed puncturing strategy achieves better reliability compared the puncturing baseline. For instance, at $P=0.01,$ the proposed puncturing strategy achieves reliability of 31\% compared to 20\% for the puncturing baseline. This means more eMBB blocks, about 50\% enhancement, are received correctly, hence less re-transmissions and better spectral efficiency. The gain of the proposed algorithm decreases while increasing the targeted eMBB SER, this is because the SER becomes dominated by the channel errors at high SER. Moreover, Fig. \ref{fig:users_rel} presents the reliability for each eMBB user. Compared to the baseline, the figure also shows that the proposed algorithm considerably enhances the reliability of the eMBB users, which means less retransmission for the eMBB traffic.
\subsection{URLLC performance}
In this section, the performance of the URLLC traffic is investigated in terms of the SER and reliability.\vspace{.02 in} 
\subsubsection{URLLC SER}
 Fig. \ref{fig:ser_urllc} illustrates the URLLC SER while varying the transmitted power. As shown in Fig. \ref{fig:ser_urllc}, the proposed puncturing scheme preserves the SER of the URLLC traffic while enhancing the SER of the eMBB traffic. The SER loss of the proposed strategy is negligible while taking into account the coding gain. We emphasize here that according to the target SER or BER for both the URLLC and eMBB, the scheduler can select either the URLLC mapper or the ESRM. 
 \subsubsection{URLLC Reliability}
 Fig.~\ref{Fig:urllc_rel} illustrates the URLLC reliability (success rate) while varying the transmitted power. The figure shows the URLLC reliability for different $\epsilon_u$. The figure shows that the proposed puncturing strategy preserves the URLLC reliability, which makes the proposed strategy a practical method for efficient multiplexing between eMBB and URLLC traffics.
\begin{figure}[t!]
	\centering
	\center{\includegraphics[width=1\columnwidth,draft=false]
		{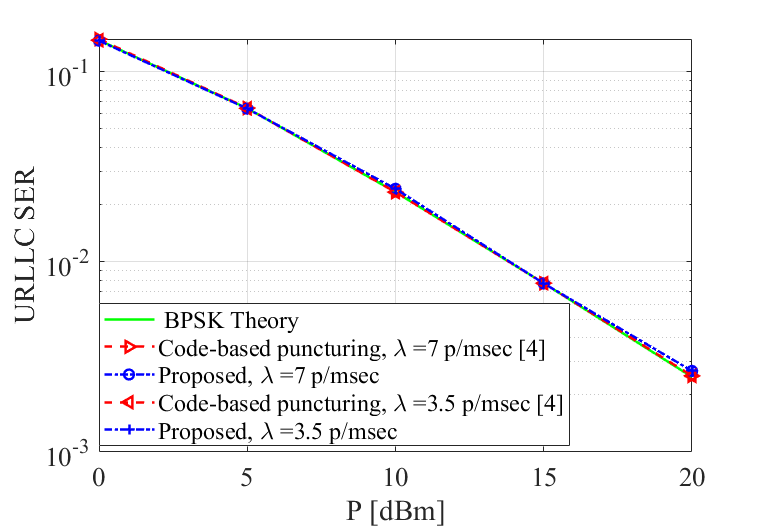}}
	\caption{URLLC SER vs transmission power in dBm.}
	\label{fig:ser_urllc}
	\vspace{-.0 in}
\end{figure} 
\begin{figure}[t!]
	\centering
	\center{\includegraphics[width=1\columnwidth,draft=false]
		{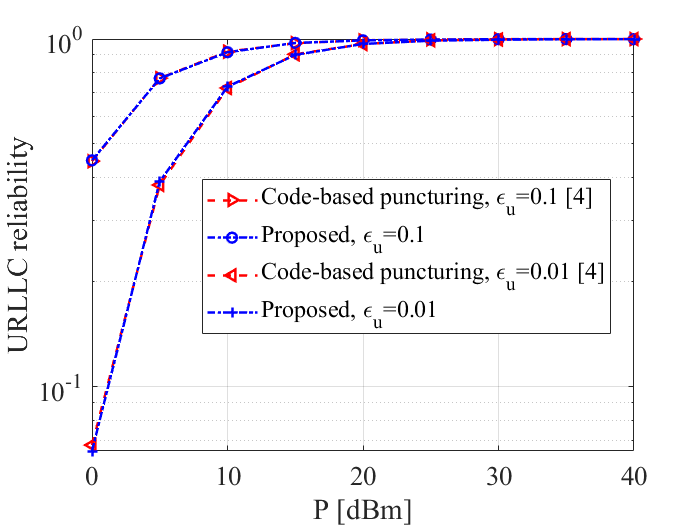}}
	\caption{URLLC reliability versus the transmitting power . $\lambda=7~p/msec$ }
	\label{Fig:urllc_rel}
	\vspace{-.0 in}
\end{figure} 

\subsection{Time Complexity }

\begin{figure}[t!]
	\centering
	\center{\includegraphics[width=1\columnwidth,draft=false]
		{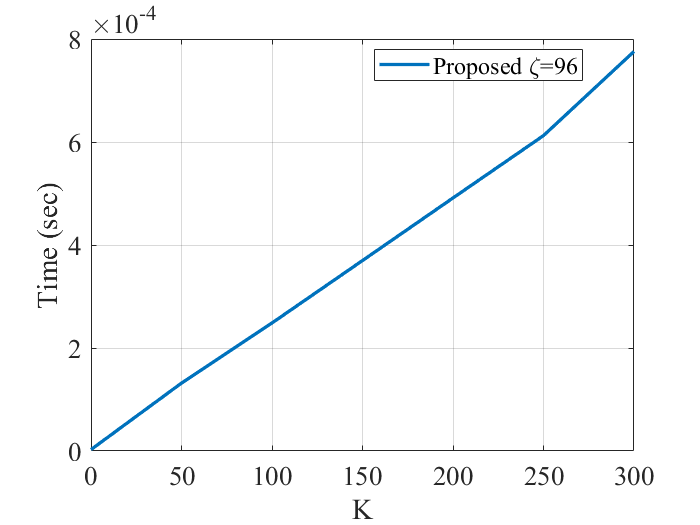}}
	\caption{Time complexity of the proposed algorithm}
	\label{Fig:complexity}
	\vspace{-.0 in}
\end{figure} 
Fig. \ref{Fig:complexity} shows the time complexity of the proposed algorithm. The algorithm was implemented in Matlab using a machine with the following characteristics: System Type: x64-based PC Processor: Intel(R) i7-8700H CPU @3.20GHz. The results show that the processing time of the proposed algorithm is less than $1~ms$. As the cloud radio access network has a powerful computational resources,  the running time of the proposed algorithm will be further reduced; Hence, the latency of the URLLC will be surely met.

 In summary, the proposed puncturing scheme represents a low complexity solution that optimizes between the SER and the spectral efficiency of the eMBB traffic while preserving the reliability of the URLLC services. The proposed algorithm gain, for the eMBB traffic, starts at 0 dB at low SNR and reach up to 10 dB at high SNR, and it enhances the eMBB reliability up to 50\%.

	\section{Concluding Remarks}
	\label{sec:conclusion}
	In this paper, we proposed a downlink puncturing strategy in an effort to reduce the impact of transmitting URLLC traffic simultaneously with eMBB traffic. The proposed strategy mitigates the impact on the eMBB traffic by exploiting the region similarity between the eMBB and URLLC symbols to reduce the effectively punctured eMBB symbols. The introduced strategy covers all relations between the eMBB and URLLC modulation schemes. Throughout the analysis, it was shown that the eMBB SER depends on the channel gain, the URLLC load, and the average similarity between the URLLC and eMBB traffic. At high SNR, the eMBB SER asymptotically saturates to the errors due to puncturing, and it is proportional to the ratio between the effectively punctured eMBB symbols to the total eMBB load. Also, when the URLLC block is small or the search space increases, the eMBB SER reduces to the errors due to the channel. Numerical and simulation results demonstrated that the proposed puncturing strategy enhances the system information rate by doubling the URLLC load for the same SER compared to the baseline. While preserving the URLLC quality of service requirements, the proposed puncturing scheme can achieve gains of up to 10 dB as compared to the baseline scheme. 
	
We believe there is still room for improving the performance of the proposed puncturing scheme.We list here a number of possible extensions that that we plan to undertake in future works. First, we plan to apply the proposed scheme to a coded system in which the eMBB and URLLC streams are coded. The code rate and the targeted block error rate should be incorporated while allocating the URLLC load. Moreover, we plan to propose a proper receiver to extract the effectively punctured eMBB symbols and decodes unaffected symbols. Second, the length of the URLLC block, $\zeta$, can be optimized to achieve better spectral efficiency by addressing the trade-off between the overhead and eMBB SER. Finally, exploiting the eMBB block with high similarity has polynomial time complexity, so it may be more efficient to build a learning model to reduce the complexity of the search algorithm. 

\appendices 
\section{Proof of Lemma 1}
Consider the Binomial distribution with $B~(\zeta,\eta_{n,m})$ to exploit the similar $\zeta-$symbols blocks between the URLLC load and the eMBB sequence. Then, the CDF $F(k)$ is expressed as
	\begin{eqnarray} 
F(k)=\sum_{j=0}^{k}\left(\begin{array}{l}{\zeta} \\{j}\end{array}\right) \left( \eta_{n,m}\right)^j(1-\eta_{n,m})^{\zeta-j}.
\end{eqnarray}
	
The expected number of similar symbols between both the URLLC and eMBB blocks is $\mu=\eta_{n,m}\,\zeta$. Under the assumption that the eMBB blocks and the URLLC packet are $i.i.d$, and by searching over the search space ($K_{m}=L_{m}-\zeta+1$), the order statistic after arranging the random samples in an increasing order is $Y_{1} \leq Y_{2 } \leq \cdots \leq Y_{K_{m} }$. Based on the results of \cite{arnold1992first}, the pmf of $Y_{z }$ becomes for all $k=0,1, \ldots, \zeta$
	\begin{equation}\label{eq:orderstatistic}
	\begin{split}
 	 f_{z }(k)=\sum\nolimits_{r=z}^{K^{m}}\left(\begin{array}{c}{K^{m}} \\ {r}\end{array}\right)\left[\{F(k)\}^{r}\{1-F(k)\}^{K_m-r}\right.\\  	 \left.-\{F(k-1)\}^{r}\{1-F(k-1)\}^{K_{m}-r}\right].  
	\end{split}
\vspace{-.1 in}	
\end{equation}

Considering the case when the largest ordered sample has at least $k$ similar symbols. Also, considering that $L_m>>l_{n,m}$, the expected number of similar symbols can be approximated as \cite{arnold1992first}
\begin{equation}\label{eq:lemma2_proof}
	 U_{n,m,\zeta}=\sum\nolimits_{k=0}^{\zeta-1}\left[1-\{F(k)\}^{L_{m}-\zeta}\right].
\end{equation} 
 Averaging over the number of $\zeta$-blocks in $l_{n,m}$, we arrive at
 \begin{equation}
	 U_{n,m,\zeta}(l_{n,m})=\frac{1}{\lceil\frac{l_{n,m}}{\zeta}\rceil}\sum_1^{\lceil\frac{l_{n,m}}{\zeta}\rceil}\sum\nolimits_{k=0}^{\zeta-1}\left[1-\{F(k)\}^{L_{m}-\zeta}\right].
\end{equation}
This completes the proof.

	\IEEEpeerreviewmaketitle
	
	\bibliographystyle{IEEEtran}
	\bibliography{main.bib}
\end{document}